\documentclass{article}
\usepackage{kr}

\usepackage{times}
\usepackage{soul}
\usepackage{url}
\usepackage[hidelinks]{hyperref}
\usepackage[utf8]{inputenc}
\usepackage[small]{caption}
\usepackage{graphicx}
\usepackage{amsmath}
\usepackage{amsthm}
\usepackage{booktabs}
\usepackage{algorithm}
\usepackage{algorithmic}
\usepackage[normalem]{ulem}
\usepackage{xcolor}
\usepackage{amsthm}
\usepackage{amssymb}
\usepackage{bussproofs}
\usepackage{comment}
\usepackage[normalem]{ulem}
\usepackage{xspace}
\usepackage{float}
\usepackage{dblfloatfix}
\usepackage{tikz}
\usetikzlibrary{arrows, automata, positioning}

\newtheorem{innercustomlem}{Lemma}
\newenvironment{customlem}[1]
  {\renewcommand\theinnercustomlem{#1}\innercustomlem}
  {\endinnercustomlem}
  
  \newtheorem{innercustomthm}{Theorem}
\newenvironment{customthm}[1]
  {\renewcommand\theinnercustomthm{#1}\innercustomthm}
  {\endinnercustomthm}

\DeclareSymbolFont{extraup}{U}{zavm}{m}{n}
\DeclareMathSymbol{\vardiamond}{\mathalpha}{extraup}{87}

\newcommand{\probbox}[1]{\centerline{\framebox{\parbox{0.95\columnwidth}{#1}}}}
\def\phi{\varphi}

\newcommand{\iffi}{\textit{iff} }
\newcommand{\dfn}{Definition}
\newcommand{\fig}{Figure}
\newcommand{\lem}{Lemma}
\newcommand{\thm}{Theorem}

\newcommand{\sect}{Section}

\newtheorem{theorem}{Theorem}

\newtheorem{lemma}[theorem]{Lemma}
\newtheorem{observation}[theorem]{Observation}

\theoremstyle{definition}
\newtheorem{definition}[theorem]{Definition}
\newtheorem{example}[theorem]{Example}
\theoremstyle{remark}

\definecolor{tim}{RGB}{0, 0, 250}
\definecolor{pio}{RGB}{200, 0, 200}
\definecolor{orange-red}{RGB}{210, 47, 0}
\definecolor{dark-green}{RGB}{51, 102, 0}

\renewcommand{\vec}[1]{\boldsymbol{#1}}
\newcommand{\var}{\mathbf{V}}

\newcommand{\const}{\mathbf{C}}
\newcommand{\ter}{\mathbf{T}}
\newcommand{\lang}{\mathcal{L}}
\newcommand{\pred}{\mathbf{P}}
\newcommand{\ar}[1]{ar(#1)}
\newcommand{\ru}{\rho}
\newcommand{\head}{\mathrm{head}}
\newcommand{\body}{\mathrm{body}}
\newcommand{\rset}{\mathcal{R}}

\newcommand{\ant}{\Gamma}

\newcommand{\con}{\Delta}

\newcommand{\sar}{\vdash}

\newcommand{\id}{(id)}

\newcommand{\conl}{(\land_{L})}
\newcommand{\conr}{(\land_{R})}
\newcommand{\negl}{(\neg_{L})}
\newcommand{\negr}{(\neg_{R})}
\newcommand{\existsl}{(\exists_{L})}
\newcommand{\existsr}{(\exists_{R})}

\newcommand{\ser}{er}
\newcommand{\eru}{\rho}
\newcommand{\seru}{s(\eru)}
\newcommand{\gt}{\mathsf{G3}}
\newcommand{\calc}{\mathsf{G3}(\rset)}
\newcommand{\prf}{\Pi}

\newcommand{\prove}{\mathtt{Prove}}
\newcommand{\true}{\mathtt{True}}
\newcommand{\false}{\mathtt{False}}

\newcommand{\Vector}[1]{\ensuremath{\vec{#1}}\xspace}

\newcommand{\pair}[1]{\langle#1\rangle}
\newcommand{\Vx}{\Vector{x}}
\newcommand{\Vy}{\Vector{y}}
\newcommand{\Vz}{\Vector{z}}

\newcommand{\Vt}{\Vector{t}}

\newcommand{\db}{\mathcal{D}}
\newcommand{\kb}{\mathcal{K}}
\newcommand{\query}{q}
\newcommand{\trig}{\tau}
\newcommand{\ch}{\mathbf{Ch}}

\newcommand{\inst}{\mathcal{I}}

\newcommand{\instb}{\mathcal{J}}

\newcommand{\homm}{\pi}

\newcommand{\sub}{\sigma}
\newcommand{\assign}{\mu}

\newcommand{\mfy}[1]{#1^{\top}}
\newcommand{\fint}{f}

\title{Connecting Proof Theory and Knowledge Representation:\\
Sequent Calculi and the Chase with Existential Rules} 

\author{%
Tim S. Lyon$^{1}$\and
Piotr Ostropolski-Nalewaja$^{1, 2}$ \\
\affiliations
1. Technische Universit\"at Dresden\\ 2. University of Wrocław\\
\emails
\{timothy\_stephen.lyon, piotr.ostropolski-nalewaja\}@tu-dresden.de
}

\begin{document}

\maketitle

\begin{abstract}
Chase algorithms are indispensable in the domain of knowledge base querying, which enable the extraction of implicit knowledge from a given database via applications of rules from a given ontology. Such algorithms have proved beneficial in identifying logical languages which admit decidable query entailment. Within the discipline of proof theory, sequent calculi have been used to write and design proof-search algorithms to identify decidable classes of logics. In this paper, we show that the chase mechanism in the context of existential rules is in essence the same as proof-search in an extension of Gentzen's sequent calculus for first-order logic. Moreover, we show that proof-search generates universal models of knowledge bases, a feature also exhibited by the chase. Thus, we formally connect a central tool for establishing decidability proof-theoretically with a central decidability tool in the context of knowledge representation.
\end{abstract}


\section{Introduction}

{\bf Existential Rules and the Chase.} The formalism of existential rules is a significant sub-discipline within the field of knowledge representation, offering insightful results within the domain of ontology-based query answering~\cite{BagLecMugSal09}, data exchange and integration~\cite{FagKolMilPop05}, and serving a central role in the study of generic decidability criteria~\cite{FelLyoOstRud23}.\footnote{Existential rules are also referred to as a \emph{tuple-generating dependencies}~\cite{AbiteboulHV95}, \emph{conceptual graph rules}~\cite{SalMug96}, Datalog$^\pm$~\cite{Gottlob09}, and \emph{$\forall \exists$-rules}~\cite{BagLecMugSal11} in the literature.} Ontology-based query answering is one of the principal problems studied within the context of existential rules, and asks if a query is logically entailed by a given knowledge base (KB) $\mathcal{K} = (\db,\rset)$, where $\db$ is a database and $\rset$ is a finite set of existential rules~\cite{BagLecMugSal11}. Databases generally consist of positive atomic facts such as $Female(Marie)$ or $Mother(Zuza, Marie)$, while existential rules---which are first-order formulae of the form $\forall{\Vx \Vy} \beta(\Vx, \Vy) \to \exists{\Vz} \alpha(\Vy, \Vz)$ with $\beta$ and $\alpha$ conjunctions of atoms---are used to encode a logical theory or ontology that permits the extraction of implicit knowledge from the encompassing KB. 

The primary tool for studying query answering within this setting is the so-called \emph{chase}, an algorithm that iteratively saturates a given database under applications of existential rules~\cite{BeeVar84}. The chase is useful in that it generates a \emph{universal model} satisfying exactly those queries entailed by a KB, and thus, allows for the reduction of query entailment to query checking over the constructed universal model~\cite{DeuNasRem08}. In this paper, we show how the chase corresponds to 
 proof-search in an extension of Gentzen's sequent calculus, establishing a connection between a central tool in the theory of existential rules with the primary decidability tool in proof theory.

\medskip

\noindent
{\bf Sequent Calculi and Proof-Search.} Since its introduction, Gentzen's sequent formalism~\cite{Gen35a,Gen35b} has become one of the preferred proof-theoretic frameworks for the creation and study of proof calculi. A sequent is an object of the form $\ant \sar \con$ such that $\ant$ and $\con$ are finite (multi)sets of logical formulae, and a sequent calculus is a set of inference rules that operate over such. Sequent systems, and generalizations thereof, have proved beneficial in establishing (meta)logical properties with a diverse number of applications, being used to write decision algorithms~\cite{Dyc92,Sla97}, to calculate interpolants~\cite{Mae60,LyoTiuGorClo20}, and have even been applied in knowledge intergation scenerios~\cite{LyoGom22}. 

It is well-known that \emph{geometric implications}, i.e. first-order formula of the form $\forall \vec{x} (\phi \rightarrow \exists \vec{y}_{1} \psi_{1} \lor \cdots \lor \exists \vec{y}_{n} \psi_{n})$ with $\phi$ and $\psi_{i}$ conjunctions of atoms, can be converted into an inference rules in a sequent calculus~\cite[p.~24]{Sim94}. Since such formulae subsume the class of existential rules, we may leverage this insight to extend Gentzen's sequent calculus for first-order logic with such rules to carry out existential rule reasoning. When we do so, we find that sequent systems mimic existential rule reasoning and proof-search (described below) simulates the chase.

Proof-search is the central means by which decidability is obtained with a sequent calculus, and usually operates by applying the inference rules of a sequent calculus bottom-up on an input sequent with the goal of constructing a proof thereof. If a proof of the input is found, the input is confirmed to be valid, and if a proof of the input is not found, a counter-model can typically be extracted witnessing the invalidity of the input. We make the novel observation that counter-models extracted from proof-search (in the context of existential rules) are universal, being homomorphically equivalent to the universal model generated by the chase.

\medskip

\noindent
{\bf Contributions.} Our contributions in this paper are as follows: (1) We establish a strong connection between tools in the domain of existential rules with that of proof theory; in particular, we show how to transform derivations with existential rules into sequent calculus proofs and vice versa. (2) We establish a correspondence between the chase and sequent-based proof-search, and (3) we recognize that proof-search, like the chase, generates universal models for knowledge bases, which is a novel, previously unknown insight regarding the capability of sequent systems.

\medskip

\noindent
{\bf Organization.} The preliminaries are located in \sect~\ref{sec:prelims}. In \sect~\ref{sec:sequents}, we present the sequent calculus framework and write a proof-search algorithm that simulates the chase. Correspondences between existential rule reasoning and sequent-based reasoning are explicated in \sect~\ref{sec:simulations}, and in \sect~\ref{sec:conclusion}, we conclude and discuss 
 future research. We note that most proofs have been deferred to the appendix. 




\section{Preliminaries and Existential Rules}\label{sec:prelims}

\newcommand{\DisplayProofScaled}{\scalebox{0.85}{\DisplayProof}}

\begin{figure*}[t]

\begin{center}
\begin{tabular}{c c c c c}
\AxiomC{}
\RightLabel{$\id$}
\UnaryInfC{$\ant, p(\vec{t}) \sar p(\vec{t}), \con$}
\DisplayProofScaled

&

\AxiomC{$\ant \sar \phi, \con$}
\RightLabel{$\negl$}
\UnaryInfC{$\ant, \neg \phi \sar \con$}
\DisplayProofScaled

&

\AxiomC{$\ant, \phi \sar \con$}
\RightLabel{$\negr$}
\UnaryInfC{$\ant \sar \neg \phi, \con$}
\DisplayProofScaled

&

\AxiomC{$\ant, \phi, \psi \sar \con$}
\RightLabel{$\conl$}
\UnaryInfC{$\ant, \phi \land \psi \sar \con$}
\DisplayProofScaled

\end{tabular}
\end{center}

\begin{center}
\begin{tabular}{c c c}
\AxiomC{$\ant \sar \phi, \con$}
\AxiomC{$\ant \sar \psi, \con$}
\RightLabel{$\conr$}
\BinaryInfC{$\ant \sar \phi \land \psi, \con$}
\DisplayProofScaled

&

\AxiomC{$\ant, \phi(y/x) \sar \con$}
\RightLabel{$\existsl~\text{$y$ fresh}$}
\UnaryInfC{$\ant, \exists x \phi \sar \con$}
\DisplayProofScaled

&

\AxiomC{$\ant \sar \exists x \phi, \phi(t/x), \con$}
\RightLabel{$\existsr$~$t \in \ter$}
\UnaryInfC{$\ant \sar \exists x \phi, \con$}
\DisplayProofScaled
\vspace{-2mm}

\end{tabular}
\end{center}

\caption{The sequent calculus $\gt$ for first-order logic.\label{fig:sequent-calc}}

\end{figure*}
\let\DisplayProofScaled\undefined

\noindent
\textbf{Formulae and Syntax.} We let $\const$ and $\var$ be two disjoint denumerable sets of \emph{constants} and \emph{variables}. We use $a, b, c, \ldots$ to denote constants and $x, y, z, \ldots$ to denote variables. We define the set of \emph{terms} to be $\ter = \const \cup \var$, and we denote terms by $t$ and annotated versions thereof. Moreover, we let $\pred = \{p, q, r, \ldots\}$ be a denumerable set of \emph{predicates} containing denumerably many predicates of each arity $n \in \mathbb{N}$, and use $\ar{p} = n$ to denote that $p \in \pred$ is of arity $n$. An \emph{atom} is a formula of the form $p(t_{1}, \ldots, t_{n})$ such that $t_{1}, \ldots, t_{n} \in \ter$ and $\ar{p} = n$. We will often write atoms as $p(\vec{t})$ with $\vec{t} = t_{1}, \ldots, t_{n}$. The \emph{first-order language} $\lang$ is defined 
 via the following grammar in Backus–Naur form:
$$
\phi ::= p(\vec{t}) \ | \ \neg \phi \ | \ \phi \land \phi \ | \ \exists x \phi
$$
 such that $p \in \pred$, $\vec{t} \in \ter$, and $x \in \var$. We use $\phi$, $\psi$, $\chi$, $\ldots$ to denote \emph{formulae} from $\lang$, and define $\phi \lor \psi := \neg (\neg \phi \land \neg \psi)$, $\phi \rightarrow \psi := \neg \phi \lor \psi$, and $\forall x \phi := \neg \exists x \neg \phi$. The occurrence of a variable is \emph{free} in a formula $\phi$ when it does not occur within the scope of a quantifier. We let $\phi(t/x)$ represent the formula obtained by substituting the term $t$ for every free occurrence of the variable $x$ in $\phi$. We use $\Gamma$, $\Delta$, $\Sigma$, $\ldots$ to denote sets of formulae from $\lang$, let $\var(\Gamma)$ denote the set of free variables in the formulae of $\Gamma$, and let $\ter(\Gamma)$ denote the set of free variables and constants occurring in the formulae of $\Gamma$. We let $i \in [n]$ represent $1 \leq i \leq n$, and define a \emph{ground atom} to be an atom $p(t_{1}, \ldots, t_{n})$ such that for each $i \in [n]$, $t_{i} \in \const$. An \emph{instance} $\inst$ is defined to be a (potentially infinite) set of atoms, and a \emph{database} $\db$ is defined to be a finite set of ground atoms. We let $\top$ be a special unary predicate and define $\mfy{\inst} = \inst \cup \{\top(c) \ | \ c \in \const\}$. An instance $\inst$ is referred to as an \emph{interpretation} \iffi $\mfy{\inst} = \inst$.
 
\medskip

\noindent
\textbf{Substitutions.} A \emph{substitution} $\sub$ is defined to be a partial function over~$\mathbf{T}$. A \emph{homomorphism} from an instance $\inst$ to an instance $\instb$ is a substitution $\homm$ from the terms of $\inst$ to the terms of~$\instb$ such that (1) if $p(t_1, \ldots, t_n) \in \inst$, then $p(\homm(t_1), \ldots, \homm(t_n)) \in \instb$, and (2) $\homm(a) = a$, for each~$a \in \const$. We say that an instance $\inst$ \emph{homomorphically maps} into an instance $\instb$ \iffi a homomorphism exists from $\inst$ to $\instb$. Two instances $\inst$ and $\instb$ are defined to be \emph{homomorphically equivalent}, written $\inst \equiv \instb$, \iffi each instance can be homomorphically mapped into the other. An \emph{$\inst$-assignment} is defined to be a substitution $\assign$ such that (1) $\assign(x) \in \ter(\inst)$, for each $x \in \var$, and (2) $\assign(a) = a$, for each~$a \in \const$. 
 For an $\inst$-assignment $\assign$, we let $\assign(\phi)$ denote the formula obtained by replacing each free variable of $\phi$ with its value under $\assign$, and we let $\assign[\vec{t}/\vec{x}]$ be the same as $\assign$, but where the variables $\vec{x}$ are respectively mapped to $\vec{t} \in \ter$.

\medskip

\noindent
\textbf{Models and Satisfaction.} Given an interpretation $\inst$ and an $\inst$-assignment $\assign$, we recursively define satisfaction $\models$ as:
\begin{flushleft}
(1) $\inst, \assign \models p(t_{1}, \ldots,t_{n})$ \iffi $p(\assign(t_{1}), \ldots, \assign(t_{n})) \in \inst$;\\
(2) $\inst, \assign \models \neg \phi$ \iffi $\inst, \assign \not\models \phi$;\\
(3) $\inst, \assign \models \phi \land \psi$ \iffi $\inst, \assign \models \phi$ and $\inst, \assign \models \psi$;\\
(4) $\inst, \assign \models \exists x \phi$ \iffi $t \in \ter(\inst)$ exists and $\inst, \assign[t/x] \models \phi$.
\end{flushleft}
 We say that $\inst$ is a \emph{model} of $\Gamma$ and write $\inst \models \Gamma$ \iffi for every $\phi \in \Gamma$ and $\inst$-assignment $\assign$, we have $\inst, \assign \models \phi$. We define an instance $\inst$ to be a \emph{universal model} of $\Gamma$ \iffi for any model $\instb$ of $\Gamma$ there exists a homomorphism from $\inst$ to $\instb$.

\medskip

\noindent
\textbf{Existential Rules.} An \emph{existential rule} is a first-order formula $\rho = \forall{\Vx\Vy}\; \beta(\Vx, \Vy) \to \exists{\Vz}\; \alpha(\Vy, \Vz)$ such that $\beta(\Vx, \Vy) = \body(\rho)$ (called the body) and $\alpha(\Vy, \Vz) = \head(\rho)$ (called the head) are conjunctions of atoms over constants and the variables $\Vx, \Vy$ and $\Vy, \Vz$, respectively. We call a finite set $\rset$ of existential rules a \emph{rule set}. We define $\Gamma$ to be \emph{$\rset$-valid} \iffi for every interpretation $\inst$, if $\inst \models \rset$, then $\inst \models \Gamma$.

\medskip

\noindent
\textbf{Derivations and the Chase.} We say that an existential rule $\rho$ is \emph{applicable} to an instance $\inst$ \iffi there exists an $\inst$-assignment $\assign$ such that $\assign(\beta(\Vx, \Vy)) \subseteq \inst$, and when this is the case, we say that $\trig = (\ru, \assign)$ is a \emph{trigger} in $\inst$. Given a trigger $\trig = (\ru, \assign)$ in $\inst$ we define an \emph{application} of the trigger $\trig$ to the instance $\inst$ to be the instance $\trig(\inst) = \inst \cup \alpha(\assign(\Vy), \Vz)$ where 
 $\Vz$ is a tuple of fresh variables. We define a \emph{chase derivation} $(\inst_i, \trig_i)_{i \in [n]}$ to be a sequence $(\inst_{1}, \trig_{1}), \ldots, (\inst_{n}, \trig_{n}), (\inst_{n+1}, \emptyset)$ such that for every $i \in [n]$, $\trig_{i}$ is a trigger in $\inst_{i}$ and $\trig_i(\inst_i) = \inst_{i+1}$. For an instance $\inst$ and a rule set $\rset$, we define the \emph{one-step chase} to be:
$$
\ch_{1}(\inst,\rset) = \!\!\!\!\!\!\!\! \bigcup_{\text{$\trig$ is a trigger in $\inst$}} \!\!\!\!\!\!\!\! \trig(\inst).
$$
 We let $\ch_{0}(\inst,\rset) = \inst$ as well as let $\ch_{n+1}(\inst,\rset) = \ch_{1}(\ch_{n}(\inst, \rset), \rset)$. Finally, we define the \emph{chase} to be $\ch_{\infty}(\inst,\rset) = \mfy{(\bigcup_{i \in \mathbb{N}} \ch_{i}(\inst,\rset))}$, which serves as a universal model of $\inst \cup \rset$~\cite{DeuNasRem08}.\footnote{We use a {\em restricted} variant of the chase; cf.~\cite{FagKolMilPop05}.}
 
\medskip

\noindent
\textbf{Queries and Entailment.} A \emph{Boolean conjunctive query (or, BCQ)} is a formula $\exists \vec{x} \query(\vec{x})$ such that $\query(\Vx)$ is a conjunction of atoms over the variables $\Vx$ and constants. We define a \emph{knowledge base (or, KB)} to be an ordered pair $\kb = (\db,\rset)$ with $\db$ a database and $\rset$ a rule set, and let $\inst$ be a \emph{model} of $\kb$, written $\inst \models \kb$, \iffi $\inst \models \db \cup \rset$. We write $\kb \models \exists \vec{x} \query(\vec{x})$ to mean that for every $\inst$, if $\inst \models \kb$, then $\inst \models \exists \vec{x} \query(\vec{x})$.  A chase derivation $(\inst_i, \trig_i)_{i \in [n]}$ \emph{witnesses} $(\db, \rset) \models \exists \vec{x} \query(\vec{x})$ \iffi $\inst_{1} = \db$, only rules from $\rset$ are applied, and there exists an $\inst_{n+1}$-assignment $\assign$ such that $\assign(\query(\vec{x})) \subseteq \inst_{n+1}$.\\ 


\section{Sequent Systems and Proof-Search}\label{sec:sequents}

We define a \emph{sequent} to be an object of the form $\ant \sar \con$ such that $\ant$ and $\con$ are \emph{finite} sets of formulae from $\lang$. 
 Typically, multisets are used in sequents rather than sets, however, we are permitted to use sets in the setting of classical logic; cf.~\cite{Kle52}. For a sequent $\ant \sar \con$, we call $\ant$ the \emph{antecedent} and $\con$ the \emph{consequent}. We define the \emph{formula interpretation} of a sequent to be $\fint(\ant \sar \con) = \bigwedge \Gamma \rightarrow \bigvee \Delta$. 

 The sequent calculus $\gt$ \cite{Kle52} for first-order logic is defined to be the set of inference rules presented in \fig~\ref{fig:sequent-calc}. It consists of the \emph{initial rule} $\id$ along with \emph{logical rules} that introduce complex logical formulae in either the antecedent or consequent of a sequent. The $\existsl$ rule is subject to a side condition, stating that the rule is applicable only if $y$ is \emph{fresh}, i.e. $y$ does not occur in the surrounding context $\Gamma,\Delta$. The $\existsr$ rule allows for the bottom-up instantiation of an existentially quantified formula with a term $t$. An \emph{application} of a rule is obtained by instantiating the rule with formulae from $\lang$. We call an application of rule \emph{top-down} (\emph{bottom-up}) whenever the conclusion (premises) is (are) obtained from the premises (conclusion).

 It is well-known that every \emph{geometric implication}, which is a formula of the form $\forall \vec{x} (\phi \rightarrow \exists \vec{y}_{1} \psi_{1} \lor \cdots \lor \exists \vec{y}_{n} \psi_{n})$ with $\phi$ and $\psi_{i}$ conjunctions of atoms, can be converted into an inference rule; see~\cite[p.~24]{Sim94} for a discussion. We leverage this insight to transform existential rules (which are special instances of geometric implications) into inference rules that can be added to the sequent calculus $\gt$. For an existential rule $\rho = \forall \vec{x} \vec{y} \beta(\vec{x},\vec{y}) \rightarrow \exists \vec{z} \alpha(\vec{y},\vec{z})$, we define its corresponding \emph{sequent rule} $s(\rho)$ to be:
\begin{center}
\small
\AxiomC{$\ant, \beta(\vec{x},\vec{y}), \alpha(\vec{y},\vec{z}) \sar \con$}
\RightLabel{$\seru$~\text{$\vec{z}$ fresh}}
\UnaryInfC{$\ant, \beta(\vec{x},\vec{y}) \sar \con$}
\DisplayProof
\end{center}
Note that we take the body $\beta(\vec{x},\vec{y})$ and head $\alpha(\vec{y},\vec{z})$ to be sets of atoms, rather than conjunctions of atoms, and we note that $\vec{x},\vec{y}$ may be instantiated with terms in rule applications. Also, $\seru$ is subject to the side condition that the rule is applicable only if all variables $\vec{z}$ are fresh. We define the sequent calculus $\calc = \gt \cup \{s(\rho) \ | \ \rho \in \rset\}$. We define a {\em derivation} to be any sequence of applications of rules in $\calc$ to arbitrary sequents, define an \emph{$\rset$-derivation} to be a derivation that only applies $\seru$ rules, and define a \emph{proof} to be a derivation starting from applications of the $\id$ rule. An example of a proof is shown on the left side of Figure~\ref{fig:simulation-example}. 

\begin{theorem}[Soundness and Completeness]\label{thm:sound-comp}
$\fint(\ant \sar \con)$ is $\rset$-valid \iffi there exists a proof of $\ant \sar \con$ in $\calc$.
\end{theorem}

\begin{figure}
\small
\centering
\scalebox{.9}{

\probbox{
	\textbf{Algorithm:} $\prove$\\
	\textbf{Input:} A sequent $\ant \sar \con$.\\
	\textbf{Output:} A Boolean $\true$ and $\false$.\\[1ex]
\textbf{If} $\ant \sar \con$ is saturated, \textbf{Return} $\false$;\\
\textbf{If} there exists a $p(\vec{t}) \in \ant \cap \con$, \textbf{Return} $\true$;\\
\textbf{If} $\phi \land \psi \in \con$, but $\phi, \psi \not\in \Delta$,\\
\hbox{}\hspace{5mm} \textbf{Set} $\con_{1} := \phi, \con$ and $\con_{2} := \psi, \con$;\\
\hbox{}\hspace{5mm} \textbf{If} $\prove(\ant \sar \con_{i}) = \false$ for some $i \in \{1,2\}$,\\
\hbox{}\hspace{5mm}\hbox{}\hspace{5mm} \textbf{Return} $\false$;\\
\hbox{}\hspace{5mm} \textbf{Else}\\
\hbox{}\hspace{5mm}\hbox{}\hspace{5mm} \textbf{Return} $\true$;\\
\textbf{If} $\exists x \phi \in \con$ and $t \in \ter(\ant)$, but $\phi(t/x) \not\in \con$,\\
\hbox{}\hspace{5mm} \textbf{Set} $\con := \phi(t/x), \con$;\\
\hbox{}\hspace{5mm} \textbf{Return} $\prove(\ant \sar \con)$\\
\textbf{Let} $\ru =  \forall{\Vx\Vy}\; \beta(\Vx, \Vy) \to \exists{\Vz}\; \alpha(\Vy, \Vz)$ be the next rule accord-\\
\hbox{}\hspace{5mm} ing to $\prec$ (if no rule has yet been picked, choose one in $\rset$);\\
\textbf{If} a $\ant$-assignment $\assign$ exists such that $\assign(\beta(\Vx, \Vy)) \subseteq \ant$, but no\\
\hbox{}\hspace{5mm} terms $\vec{t} \in \ter(\ant)$ exist such that $\assign[\vec{t}/\vec{z}](\alpha(\Vy, \Vz)) \subseteq \ant$;\\
\hbox{}\hspace{5mm} \textbf{Set} $\ant := \alpha(\assign(\Vy),\Vz), \ant$ with $\Vz$ fresh;\\
\hbox{}\hspace{5mm} \textbf{Return} $\prove(\ant \sar \con)$;
}}
\caption{The proof-search algorithm $\prove$.\label{fig:proof-search}}
\end{figure}

\newcommand{\vsep}[1]{\resizebox{0.1pt}{#1}{\phantom{.}}}

\begin{figure*}
\scalebox{0.8}{
\begin{minipage}{.55\textwidth}
\AxiomC{}
\RightLabel{$\id$}
\UnaryInfC{$\ant \sar \exists x(\mathtt{A}(x,a) \land \mathtt{F}(x)), \mathtt{A}(c,a)$} 

\AxiomC{}
\RightLabel{$\id$}
\UnaryInfC{$\ant \sar \exists x(\mathtt{A}(x,a) \land \mathtt{F}(x)), \mathtt{F}(c)$} 

\RightLabel{$\conr$}
\BinaryInfC{$\ant \sar \exists x(\mathtt{A}(x,a) \land \mathtt{F}(x)), \mathtt{A}(c,a) \land \mathtt{F}(c)$} 
\RightLabel{$\existsr$}
\UnaryInfC{$\mathtt{M}(b,a), \mathtt{A}(b,a), \mathtt{F}(b), \mathtt{M}(c,b), \mathtt{A}(c,b), \mathtt{F}(c), \mathtt{A}(c,a) \sar \exists x(\mathtt{A}(x,a) \land \mathtt{F}(x))$}
\RightLabel{$s(\rho_{2})$}
\UnaryInfC{$\mathtt{M}(b,a), \mathtt{A}(b,a), \mathtt{F}(b), \mathtt{M}(c,b), \mathtt{A}(c,b), \mathtt{F}(c) \sar \exists x(\mathtt{A}(x,a) \land \mathtt{F}(x))$}
\RightLabel{$s(\rho_{1})$}
\UnaryInfC{$\mathtt{M}(b,a), \mathtt{A}(b,a), \mathtt{F}(b), \mathtt{M}(c,b) \sar \exists x(\mathtt{A}(x,a) \land \mathtt{F}(x))$}
\RightLabel{$s(\rho_{1})$}
\UnaryInfC{$\mathtt{M}(b,a), \mathtt{M}(c,b) \sar \exists x(\mathtt{A}(x,a) \land \mathtt{F}(x))$}
\DisplayProof
\end{minipage}

\begin{minipage}{.25\textwidth}
\ \\
\end{minipage}
\begin{minipage}{.2\textwidth}
\begin{tikzpicture}
		
		\node[] (a2) [] {$a$}; 
		\node[] (b2) [above=of a2,label={[label distance=1pt]0:$\mathtt{F}$}] {$b$};
		\node[] (z12) [above=of b2,label={[label distance=1pt]90:$\mathtt{F}$}] {$c$};
		
		\draw[->] (z12) -- (b2) node [midway,xshift=-0.5em,label={[label distance=1pt]0:$\mathtt{M}$}] {};
		\draw[->] (b2) -- (a2) node [midway,xshift=-0.5em,label={[label distance=1pt]0:$\mathtt{M}$}] {};
		
    \draw [->](b2) to [out=225,in=135,looseness=1] node[midway,label={[label distance=0pt]180:$\mathtt{A}$}] {} (a2);
    
    \draw [->](z12) to [out=225,in=135,looseness=1] node[midway,label={[label distance=0pt]180:$\mathtt{A}$}] {} (b2);
    
    \draw [->](z12) to [out=-45,in=45,looseness=1] node[midway,label={[label distance=0pt]0:$\mathtt{A}$}] {} (a2);

		\node[] (x) [right=of z12,xshift=1cm,label={[label distance=1pt]90:$\mathtt{F}$}] {$x$};
		\node[] (a) [below=of x,yshift=-4.25em] {$a$};
		
    \draw [->](x) to [out=-45,in=45,looseness=1] node[midway,label={[label distance=0pt]0:$\mathtt{A}$}] {} (a);
    
	\draw[->,dotted] (x) -- (z12) node [midway,label={[label distance=1pt]90:$\assign$}] {};
	\draw[->,dotted] (a) -- (a2) node [midway,label={[label distance=1pt]90:$\assign$}] {};

\end{tikzpicture}
\end{minipage}
}
\caption{Above left is a proof in $\calc$ witnessing that $\kb \models \exists x(\mathtt{A}(x,a) \land \mathtt{F}(x))$, where $\kb = (\db, \rset)$ is as defined in Example~\ref{ex:example-proof-chase} and $\ant = \mathtt{M}(b,a), \mathtt{A}(b,a), \mathtt{F}(b), \mathtt{M}(c,b), \mathtt{A}(c,b), \mathtt{F}(c), \mathtt{A}(c,a)$. Above right is an illustration showing that the BCQ $\exists x(\mathtt{A}(x,a) \land \mathtt{F}(x))$ (to the right) can be mapped into the chase $\ch_{\infty}(\db,\rset)$ (to the left) via the $\ch_{\infty}(\db,\rset)$ -assignment $\assign$ (dotted arrows).\label{fig:simulation-example}}
\end{figure*}
 
We now define a proof-search algorithm that decides (under certain conditions) if a BCQ is entailed by a knowledge base. The algorithm $\prove$ (shown in \fig~\ref{fig:proof-search}) takes a sequent of the form $\db \sar \exists\vec{x} \query(\vec{x})$ as input and bottom-up 
 applies inference rules from $\calc$ with the goal of constructing a proof thereof. Either, $\prove$ returns a proof witnessing that $(\db,\rset) \models \exists\vec{x} \query(\vec{x})$, or a counter-model to this claim can be extracted from failed proof search. The functionality of this algorithm depends on certain \emph{saturation conditions}, defined in \dfn~\ref{def:saturation} below, and which determine when a rule from $\calc$ is bottom-up applicable. 
 Due to the shape of the input $\db \sar \exists\vec{x} \query(\vec{x})$, only $\id$, $\conr$, $\existsr$, and $\seru$ rules are applicable during proof search. 
Moreover, we let $\prec$ be an arbitrary cyclic order over $\rset = \{\ru_{1}, \ldots, \ru_{n}\}$, that is, $\ru_{1} \prec \ru_{2} \cdots \ru_{n-1} \prec \ru_{n} \prec \ru_{1}$. We use $\prec$ to ensure the \emph{fair application} of $\seru$ rules during proof-search, meaning that no bottom-up rule application is delayed indefinitely.

\begin{definition}[Saturation]\label{def:saturation} Let $\ant \sar \con$ be a sequent. We say that $\ant \sar \con$ is \emph{saturated} \iffi it satisfies the following:
\begin{description}

\item[$id$.] if $p(\vec{t}) \in \Gamma$, then $p(\vec{t}) \not\in \Delta$;




\item[$\land_{R}$.] If $\phi \land \psi \in \Delta$, then either $\phi\in \Delta$ or $\psi \in \Delta$;


\item[$\exists_{R}$.] If $\exists x \phi \in \Delta$, then for every $t \in \ter(\Gamma)$, $\phi(t/x) \in \Delta$;

\item[$\ser$.] For each $\ru \in \rset$, if a $\ant$-assignment $\assign$ exists such that $\assign(\body(\ru)) \subseteq \ant$, then there exist $\vec{t} \in \ter(\Gamma)$ such that $\assign[\vec{t}/\vec{z}](\head(\ru)) \subseteq \ant$ holds.

\end{description}
\end{definition}

\begin{theorem}\label{thm:correctness-proof-search} Let $\rset$ be a rule set, $\db$ be a database, and $\exists\vec{x} \query(\vec{x})$ be a BCQ. Then,
\begin{enumerate}

\item If $\prove(\db \sar \exists\vec{x} \query(\vec{x})) = \true$, then a proof in $\calc$ can be constructed witnessing that $(\db, \rset) \models \exists\vec{x} \query(\vec{x})$;

\item If $\prove(\db \sar \exists\vec{x} \query(\vec{x})) \neq \true$, then a universal model can be constructed witnessing that $(\db, \rset) \not\models \exists\vec{x} \query(\vec{x})$.

\end{enumerate}
\end{theorem}

We refer to the universal model of $(\db, \rset)$ stated in the second claim of \thm~\ref{thm:correctness-proof-search} as the \emph{witnessing counter-model}.

\section{Simulations and Equivalences}\label{sec:simulations}

 We present a sequence of results which culminate in the establishment of two main theorems: (1) \thm~\ref{thm:simulation}, which confirms that chase derivations are mutually transformable with certain proofs in $\calc$, and (2) \thm~\ref{thm:chase-prove-equiv}, which confirms the equivalence of $\prove$ and the chase. We end the section by providing an example illustrating the latter correspondence between proofs and the chase.

\begin{observation}\label{lem:con-ex-permutabove-ER}
Let $\rset$ be a rule set. If $\ru \in \rset$, then any application of $\conr$ and $\existsr$ permute above $\seru$.
\end{observation}

\begin{proof} It is straightforward to confirm the permutation of such rules as the $\seru$ rules operate on the antecedent of a sequent, and $\conr$ and $\existsr$ operate on the consequent.
\end{proof}

\begin{observation}\label{obs:X}
If $\inst$ is an instance, then only $\seru$ rules of $\calc$ can be bottom-up applied to $\inst \sar \emptyset$. Moreover, such an application yields a sequent $\inst' \sar \emptyset$ with $\inst'$ an instance.
\end{observation}

\begin{observation}\label{obs:A}
The inference shown below left is a correct application of $\seru$ \iffi the inference shown below right is:
\begin{center}
\small
\AxiomC{$\ant' \sar \emptyset$}
\RightLabel{$\seru$}
\UnaryInfC{$\ant \sar \emptyset$}
\DisplayProof
\hspace{5mm}
\AxiomC{$\ant' \sar \con$}
\RightLabel{$\seru$}
\UnaryInfC{$\ant \sar \con$}
\DisplayProof
\end{center}
\end{observation}

\begin{observation}\label{obs:B}
Let $\inst$ and $\inst'$ be instances with $\trig = (\ru, \assign)$ a trigger on $\inst$. Then, $(\inst, \trig), (\inst', \emptyset)$ is a chase derivation \iffi the following is a correct application of $\seru$:
\begin{center}
\small
\AxiomC{$\inst' \sar \emptyset$}
\RightLabel{$\seru$}
\UnaryInfC{$\inst \sar \emptyset$}
\DisplayProof
\end{center}
\end{observation}

\begin{lemma}\label{lem:mutual-transformation} 
For every rule set $\rset$, $n \in \mathbb{N}$, and instances $\inst_1, \ldots, \inst_n$ there exists a chase derivation $(\inst_i, \trig_i)_{i \in [n - 1]}$ \iffi 
there exists an $\rset$-derivation of $\inst_{1} \sar \emptyset$ from $\inst_{n} \sar \emptyset$.
\end{lemma}
%


In the proof of the following theorem, one shows that every chase derivation can be transformed into a proof in $\calc$ and vice-versa, showing how existential rule reasoning and proofs in $\calc$ simulate one another.

\begin{theorem}\label{thm:simulation} Let $\rset$ be a rule set. A chase derivation $(\inst_i, \trig_i)_{i \in [n]}$ witnessing $(\db, \rset) \models \exists \vec{x} \query(\vec{x})$ exists \iffi a proof in $\calc$ of $\db \sar \exists \vec{x} \query(\vec{x})$ exists.
\end{theorem}

 Leveraging Theorems~\ref{thm:correctness-proof-search} and~\ref{thm:simulation}, it is straightforward to prove the first claim of the theorem below. The second claim is immediate as $\inst$ and $\ch_{\infty}(\db,\rset)$ are universal models. We note that the following theorem expresses a correspondence between proof-search and the chase.

\begin{theorem}\label{thm:chase-prove-equiv}
Let $\rset$ be a rule set, $\db$ be a database, and $\exists \vec{x} q(\vec{x})$ be a BCQ. Then,
\begin{enumerate}

\item $\prove(\db \sar \exists \vec{x} q(\vec{x})) = \true$ \iffi there is an $n \in \mathbb{N}$ such that $\ch_{n}(\db,\rset) \models \exists \vec{x} q(\vec{x})$ \iffi $\ch_{\infty}(\db,\rset) \models \exists \vec{x} q(\vec{x})$;

\item If $\prove(\db \sar \exists \vec{x} q(\vec{x})) \neq \true$ , then $\inst \equiv \ch_{\infty}(\db,\rset)$ with $\inst$ the witnessing counter-model.

\end{enumerate}
\end{theorem}

\begin{example}\label{ex:example-proof-chase} We provide an example demonstrating the relationship between a proof and the chase. We read $\mathtt{F}(x)$ as `$x$ is female', $\mathtt{M}(x,y)$ as `$x$ is the mother of $y$' and $\mathtt{A}(x,y)$ as `$x$ is the ancestor of $y$'. We let $\kb = (\db, \rset)$ be a knowledge base such that $\db = \{\mathtt{M}(b,a),\mathtt{M}(c,b)\}$, $\rset = \{\ru_{1}, \ru_{2}\}$, and
\begin{flushleft}
$\rho_{1} = \forall xy (\mathtt{M}(x,y) \rightarrow \mathtt{A}(x,y) \land \mathtt{F}(x))$;\\
$\rho_{2} = \forall xy (\mathtt{A}(x,y) \land \mathtt{A}(y,z) \rightarrow \mathtt{A}(x,z))$.
\end{flushleft}
 In \fig~\ref{fig:simulation-example}, $\kb \models \exists x(\mathtt{A}(x,a) \land \mathtt{F}(x))$ is witnessed and verified by the proof shown left. The graph shown right demonstrates that the BCQ $\exists x(\mathtt{A}(x,a) \land \mathtt{F}(x))$ (to the right) can be mapped into the chase $\ch_{\infty}(\db,\rset)$ (to the left) via a $\ch_{\infty}(\db,\rset)$-assignment $\assign$ (depicted as dotted arrows). (NB. We have omitted the points $\{\top(c) \ | \ c \in \const\}$ in the picture of $\ch_{\infty}(\db,\rset)$ for simplicity.)
\end{example}

\section{Concluding Remarks}\label{sec:conclusion}


 We have formally established an equivalence between existential rule reasoning and sequent calculus proofs, effectively showing that proof-search simulates the chase. This work is meaningful as it uncovers and connects two central reasoning tasks and tools in the domain of existential rules and proof theory. Moreover, we have found that the counter-models extracted from failed proof-search are universal, implying their homomorphic equivalence to the chase---a previously unrecognized observation. 
 
 For future work, we aim to examine the relationship between the \emph{disjunctive chase}~\cite{BouManMorPie16} and proof-search in sequent calculi with disjunctive inference rules. It may additionally be worthwhile to investigate if our sequent systems can be adapted to facilitate reasoning with non-classical variants or extensions of existential rules. For example, we could merge our sequent calculi with those of~\cite{LyoGom22} for \emph{standpoint logic}---a modal logic used in knowledge integration to reason with diverse and potentially conflicting knowledge sources~\cite{GomRud21}. Finally, as this paper presents a sequent calculus for querying with existential rules, we plan to further explore its utility; e.g. by identifying admissible rules or applying loop checking techniques to uncover new classes of existential rules with decidable query entailment.
 

\section*{Acknowledgments}

Work supported by the European Research Council (ERC) Consolidator Grant 771779 
 (DeciGUT).

\bibliographystyle{kr}
\bibliography{bibliography}

\appendix

\section{Proofs for \sect~\ref{sec:sequents}}

\begin{lemma}[Soundness]\label{lem:sound}
If there exists a proof of $\ant \sar \con$ in $\calc$, then $\fint(\ant \sar \con)$ is $\rset$-valid.
\end{lemma}

\begin{proof} We prove the result by induction on the number of inferences in the proof of $\ant \sar \con$.

\textit{Base case.} If $\ant \sar \con$ is an instance of $\id$, then $p(\vec{t}) \in \ant \cap \con$. Let $\ant = \ant', p(\vec{t})$ and $\Delta = \Delta', p(\vec{t})$. Suppose $\inst$ is an interpretation and $\assign$ is an $\inst$-assignment such that $\inst, \assign \models \bigwedge \Gamma' \land p(\vec{t})$. Then, $\assign(p(\vec{t})) \in \inst$, which shows that $\inst, \assign \models \bigvee \Delta' \lor p(\vec{t})$. Hence, $\inst \models \bigwedge \Gamma \rightarrow \bigvee \Delta$ for any interpretation $\inst$, including any interpretation $\inst$ such that $\inst \models \rset$. This shows that $\fint(\ant \sar \con)$ is $\rset$-valid.

\textit{Inductive step.} We consider the final inference in the given proof and show that if the conclusion is not $\rset$-valid, then the premise is not $\rset$-valid, establishing that the conclusion is $\rset$-valid as the premise is $\rset$-valid by the inductive hypothesis. We argue the $\existsl$, $\existsr$, and $\seru$ cases as the remaining cases are straightforward.

$\existsl$. Suppose an interpretation $\inst$ and $\inst$-assignment $\assign$ exist such that $\inst \models \rset$ and $\inst, \assign \not\models \fint(\ant, \exists x \phi \sar \con)$. Then, there exists a term $t \in \ter(\inst)$ such that $\inst, \assign[t/x] \models \phi$. It follows that $\inst, \assign[t/y] \models \phi(y/x)$, showing that $\inst, \assign[t/y] \not\models \fint(\ant, \phi(y/x) \sar \con)$, which concludes the proof of the case.

$\existsr$. Suppose an interpretation $\inst$ and $\inst$-assignment $\assign$ exist such that $\inst \models \rset$ and $\inst, \assign \not\models \fint(\ant \sar \exists x \phi, \con)$. It follows that $\inst, \assign \not\models \exists x \phi$, meaning that for every term $t \in \ter(\inst)$, $\inst, \assign \not\models \phi(t/x)$. As $\inst$ is an interpretation, $\ter(\inst) = \ter$, showing that $\inst, \assign \not\models \phi(t/x)$. This implies that $\inst, \assign \not\models \fint(\ant \sar \exists x \phi, \phi(t/x), \con)$, i.e. the premise is not $\rset$-valid.

$\seru$. Let $\rho = \forall \vec{x} \vec{y} \beta(\vec{x},\vec{y}) \rightarrow \exists \vec{z} \alpha(\vec{y},\vec{z}) \in \rset$. Suppose an interpretation $\inst$ and $\inst$-assignment $\assign$ exist such that $\inst \models \rset$ and $\inst, \assign \not\models \fint(\ant, \beta(\vec{x},\vec{y}) \sar \con)$. It follows that we have $\inst, \assign \models \beta(\vec{x},\vec{y})$, implying $\inst, \assign \models \exists \vec{z} \alpha(\vec{y},\vec{z})$ as $\inst \models \rset$. Therefore, there exist $\vec{t} \in \ter(\inst)$ such that $\inst, \assign[\vec{t}/\vec{z}] \models \alpha(\vec{y},\vec{z})$, showing that $\inst, \assign \not\models \fint(\ant, \beta(\vec{x},\vec{y}), \alpha(\vec{y},\vec{z}) \sar \con)$, i.e. the premise is not $\rset$-valid.
\end{proof}

\begin{lemma}[Completeness]\label{lem:complete}
If $\fint(\ant \sar \con)$ is $\rset$-valid, there exists a proof of $\ant \sar \con$ in $\calc$.
\end{lemma}

\begin{proof} Similar to the proof-search procedure $\prove$, we take $\ant \sar \con$ as input and apply rules from $\calc$ bottom-up with the goal of constructing a proof thereof. We bottom-up apply the rules from $\calc$ in a roundabout fashion, applying $\id$, $\negl$, $\negr$, $\conl$, $\conr$, $\existsl$, $\existsr$, $\seru$, and then circling back to $\id$. When we consider a rule, we bottom-up apply it in all possible ways to the derivation being constructed. Furthermore, if a path in the derivation exists such that it is not an instance of $\id$ and none of the other rules are bottom-up applicable to the sequent at the top of the path, then we copy it above itself. This strategy of completeness is common in the proof theory literature; e.g. see~\cite{Kri59,LyoKar22}.

Let us suppose that no proof of $\ant \sar \con$ exists in $\calc$. Then, the above described process will not terminate, implying that an infinite derivation of $\ant \sar \con$ will be constructed, which has the shape of an infinite tree with finite branching. Therefore, by K\"onig's lemma an infinite path
$$
\mathcal{P} = (\Gamma_{0} \sar \Delta_{0}), (\Gamma_{1} \sar \Delta_{1}), \ldots, (\Gamma_{n} \sar \Delta_{n}), \ldots
$$
 of sequents exists such that $\Gamma_{0} = \ant$ and $\Delta_{0} = \con$. We define:
$$
\mathbf{\ant} = \bigcup_{i \in \mathbb{N}} \ant_{i}
\qquad
\mathbf{\con} = \bigcup_{i \in \mathbb{N}} \con_{i}
$$
 Using $\mathbf{\ant}$, we define an interpretation $\inst$ accordingly:
$$
\inst = \mfy{\{p(\vec{t}) \ | \ p(\vec{t}) \in \mathbf{\ant} \}}
$$
 It is straightforward to argue that $\inst \models \rset$ and is similar to the argument given in \thm~\ref{thm:correctness-proof-search} below. Let us now define an $\inst$-assignment $\assign$ such that $\assign(t) = t$ for each term $t \in \ter$. We now show by simultaneous induction on the number of logical connectives in $\phi$ that (1) if $\phi \in \mathbf{\Gamma}$, then $\inst, \assign \models \phi$, and (2) if $\phi \in \mathbf{\Delta}$, then $\inst, \assign \not\models \phi$. We only show the atomic, negation, and existential cases below for claim (1) as the remaining cases are simple or similar.

\begin{description}

\item[$p(\vec{t}) \in \mathbf{\Gamma}$.] If $p(\vec{t}) \in \mathbf{\Gamma}$, then $p(\vec{t}) \in \inst$ by definition, showing that $\assign(p(\vec{t})) \in \inst$, which implies that $\inst, \assign \models p(\vec{t})$.

\item[$\neg \psi \in \mathbf{\Gamma}$.] If $\neg \psi \in \mathbf{\Gamma}$, then $\psi \in \mathbf{\Delta}$ as the $\negl$ rule will eventually be bottom-up applied in the procedure described above, so by IH and claim (2), $\inst, \assign \not\models \psi$, which shows that $\inst, \assign \models \neg \psi$.

\item[$\exists x \psi \in \mathbf{\Gamma}$.] If $\exists x \psi \in \mathbf{\Gamma}$, then at some stage $\existsl$ will be applied bottom-up ensuring that $\psi(y/x) \in \mathbf{\Gamma}$. By IH, $\inst, \assign[y/y] \models \psi(y/x)$, and we know $y \in \ter(\inst)$, showing that $\inst, \assign \models \exists x \psi$.

\end{description}
 Since $\Gamma \subseteq \mathbf{\Gamma}$ and $\Delta \subseteq \mathbf{\Delta}$, we have that $\inst, \assign \models \bigwedge \Gamma$ and $\inst, \assign \not\models \bigvee \Delta$, implying that $\inst, \assign \not\models \bigwedge \Gamma \rightarrow \bigvee \Delta$, which demonstrates that $\fint(\ant \sar \con)$ is not $\rset$-valid.
\end{proof}

\begin{customthm}{\ref{thm:sound-comp}}[Soundness and Completeness]
$\fint(\ant \sar \con)$ is $\rset$-valid \iffi there exists a proof of $\ant \sar \con$ in $\calc$.
\end{customthm}

\begin{proof}
Follows from \lem~\ref{lem:sound} and \lem~\ref{lem:complete} above.
\end{proof}

\begin{customthm}{\ref{thm:correctness-proof-search}} Let $\rset$ be a rule set, $\db$ be a database, and $\exists\vec{x} \query(\vec{x})$ be a BCQ. Then,
\begin{enumerate}

\item If $\prove(\db \sar \exists\vec{x} \query(\vec{x})) = \true$, then a proof in $\calc$ can be constructed witnessing that $(\db, \rset) \models \exists\vec{x} \query(\vec{x})$;

\item If $\prove(\db \sar \exists\vec{x} \query(\vec{x})) \neq \true$, then a universal model can be constructed witnessing that $(\db, \rset) \not\models \exists\vec{x} \query(\vec{x})$.

\end{enumerate}
\end{customthm}

\begin{proof} The first claim is immediate since if $\prove(\db \sar \exists\vec{x} \query(\vec{x})) = \true$, then $\prove$ constructs a proof of $\db \sar \exists\vec{x} \query(\vec{x})$ as every recursive call of $\prove$ corresponds to a bottom-up application of $\conr$, $\existsr$, or $\seru$. This implies that $(\db, \rset) \models  \exists\vec{x} \query(\vec{x})$ as $\calc$ is sound. Let us therefore argue that the second claim holds. 

We assume w.l.o.g. that $\prove$ does not terminate and show how to extract a counter-model witnessing that $(\db, \rset) \not\models \exists\vec{x} \query(\vec{x})$. Since $\prove$ does not terminate, it generates an infinite derivation in the form of a tree with $\db \sar \exists\vec{x} \query(\vec{x})$ root and which is finite branching (as $\calc$ only consists of unary and binary rules). Hence, by K\"onig's lemma there exists an infinite path
$$
\mathcal{P} = (\Gamma_{0} \sar \Delta_{0}), (\Gamma_{1} \sar \Delta_{1}), \ldots, (\Gamma_{n} \sar \Delta_{n}), \ldots
$$
 of sequents in the infinite derivation such that $\Gamma_{0} = \db$ and $\Delta_{0} = \exists\vec{x} \query(\vec{x})$. We use this path to construct an interpretation $\inst$ such that $\inst \models \db$ and $\inst \not\models \exists\vec{x} \query(\vec{x})$. Let us now define: 
$$
\inst = \mfy{(\bigcup_{i \in \mathbb{N}} \Gamma_{i})} \quad 
\Delta = \bigcup_{i \in \mathbb{N}} \Delta_{i}
$$
 We now argue (1) $\inst \models (\db,\rset)$, and (2) if $\phi \in \Delta$, then $\inst \not\models \phi$. We argue claim (1) first: Since $\db = \Gamma_{0} \subseteq \inst$, we know that $\inst, \assign \models \db$ for any $\inst$-assignment $\assign$ as all $\inst$-assignments map constants in the same way and $\db$ contains only constants; hence, $\inst \models \db$. Let us now argue that $\inst \models \rset$ as well. Let $\assign$ be an arbitrary $\inst$-assignment and $\ru \in \rset$ with $\ru = \forall{\Vx\Vy}\; \beta(\Vx, \Vy) \to \exists{\Vz}\; \alpha(\Vy, \Vz)$. Suppose that $\inst, \assign \models \beta(\Vx, \Vy)$. It follows that for some $\Gamma_{i} \sar \Delta_{i}$ in $\mathcal{P}$, $\assign(\beta(\Vx, \Vy)) \subseteq \Gamma_{i}$. If $\Gamma_{i}, \assign \not\models \exists{\Vz}\; \alpha(\Vy, \Vz)$, then due to the cyclic order $\prec$ (ensuring fairness) imposed during proof-search, eventually $\ru$ will be considered and $\seru$ applied bottom-up. This ensures that $\alpha(\assign(\Vy), \Vz) \in \inst$ with $\Vz$ fresh in the application of $\seru$. Hence, $\inst, \assign \models \ru$, showing that $\inst \models \ru$ as $\assign$ was assumed arbitrary, establishing claim (1).
 
 Let us now argue (2) by induction on the number of logical operators in $\phi$. Note that due to the shape of the input $\db \sar \exists\vec{x} \query(\vec{x})$ and the rules applied during proof-search, only atomic formulae, conjunctions, and existentials will occur in $\Delta$. We define $\assign(t) = t$ for each $t \in \ter$. 
 
\begin{description}

\item[$p(\vec{t}) \in \Delta$.] If $p(\vec{t}) \in \inst$, then there exists some $i$ such that $p(\vec{t}) \in \Gamma_{i} \cap \Delta_{i}$, meaning that $\prove$ would terminate and return $\true$ contrary to our assumption. Thus, $p(\vec{t}) \not\in \inst$, implying that $\inst, \assign \not\models p(\vec{t})$.\\

\item[$\psi \land \chi \in \Delta$.] If $\psi \land \chi \in \Delta$, then there exists a minimal $i$ such that $\psi \land \chi \in \Delta_{i}$. Therefore, by the conjunction step of $\prove$, we know that at some stage $j \geq i$ either $\psi \in \Delta_{j}$ or $\chi \in \Delta_{j}$. By IH, either $\psi \in \Delta$ or $\chi \in \Delta$, meaning either $\inst, \assign \not\models \psi$ or $\inst, \assign \not\models \chi$, showing that $\inst, \assign \not\models \psi \land \chi$ regardless of which case holds.\\

\item[$\exists x \psi \in \Delta$.] Let $t \in \ter(\inst)$. Since $\exists x \psi \in \Delta$, there exists some minimal $i$, $\exists x \psi \in \Delta_{i}$. By the existential step of $\prove$, we know that $\psi(t/x) \in \Delta_{j}$ for some $j \geq i$. By IH, $\inst \not\models \psi(t/x)$, and since $t$ was chosen arbitrarily, we have that  for all $t \in \ter(\inst)$, $\inst, \assign \not\models \psi(t/x)$, showing that $\inst, \assign \not\models \exists x \psi$.

\end{description}
 This concludes the proof of claim (2). As a consequence, since $\exists \vec{x} q(\vec{x}) \in \Delta$, it follows that $\inst, \assign \not\models \exists \vec{x} q(\vec{x})$. This fact, in conjunction with claim (1), establishes that $(\db, \rset) \not\models \exists \vec{x} q(\vec{x})$.
\newcommand{\nats}{\mathbb{N}} 
 Last, we argue that $\inst$ is a universal model for $(\db,\rset)$. Let $\instb$ be any model of $(\db, \rset)$. We first define a sequence $s' = \pair{\Gamma_i'}_{i \in \nats}$ relative to the sequence $s = \pair{\Gamma_i}_{i \in \nats}$ of antecedents from the path $\mathcal{P}$ accordingly: (1) $\ant_{0}' = \ant_{0} = \db$, (2) $\ant_{n+1}' = \min\{\ant_{k} \mid \ant_{k} \neq \ant_{n}', n \leq k \}$. 
 
 We inductively build an increasing sequence of homomorphisms $\pair{h_i : \Gamma_{i}' \to \instb}_{i \in \nats}$. By the definition of $s'$, we have that $\Gamma_{0}' = \Gamma_0 = \db$. As $\db \subseteq \instb$, we simply define $h_0$ to be an identity on the domain of $\db$. Now, assume that $h_i$ is a homomorphism from $\Gamma_{i}' \to \instb$. We will build a homomorphism $h_{i + 1} : \Gamma_{i + 1}' \to \instb$. 
 Observe, $\Gamma_{i + 1}'$ can only be obtained from $\Gamma_{i}'$ by an application of $\seru$ with $\ru = \forall{\Vx\Vy}\; \beta(\Vx, \Vy) \to \exists{\Vz}\; \alpha(\Vy, \Vz)$. Thus, from Observation~\ref{obs:B}, a trigger $\tau = (\eru, \mu)$ exists in $\Gamma_{i}'$. Let $\mu(\Vx,\Vy) = \Vt,\Vt'$ where $\Vx,\Vy$ is the tuple of universally quantified variables in $\body(\eru)$. Note that there exists a trigger $(\eru, \mu')$ in $\instb$ such that $\mu'(\Vx,\Vy) = h_i(\Vt,\Vt')$. However, as $\instb$ is a model of $(\db, \rset)$ we know that there exists a $\instb$-assignment $\mu''$ mapping $\head(\eru)$ into $\instb$ that agrees on $\Vy$ with $\mu'$. As $\Gamma_{i + 1}'$ is created by applying the trigger $\tau$ to $\Gamma_{i}'$, we know that there exists a $\ant_{i}'$-assignment $\assign'''$ that maps $\head(\eru)$ to $\Gamma_{i + 1}'$ which agrees with $\assign$ on $\Vy$. From those facts we can see that $h_{i + 1}$ can be constructed from $h_{i}$ by ensuring that $h_{i + 1} \circ \assign'''$ is identical to $\assign''$ when restricted to $\Vz$. Note that $h_{i + 1} \circ \assign'''$ already agrees on $\Vx$ and $\Vy$ with $\assign''$. Finally, as the homomorphisms $h_i$ form an ascending sequence, we take the limit $\bigcup_{i \in \nats}h_i$ as the homomorphism from $\inst$ to $\instb$. 
\end{proof}

\section{Proofs for \sect~\ref{sec:simulations}}

\begin{customlem}{\ref{lem:mutual-transformation}}
For every rule set $\rset$, $n \in \mathbb{N}$, and instances $\inst_1, \ldots, \inst_n$ there exists a chase derivation $(\inst_i, \trig_i)_{i \in [n - 1]}$ \iffi 
there exists an $\rset$-derivation of $\inst_{1} \sar \emptyset$ from $\inst_{n} \sar \emptyset$.
\end{customlem}
\begin{proof}
$(\Rightarrow)$ Follows from Observation~\ref{obs:B}.
$(\Leftarrow)$ Let $\prf$ be a bottom-up $\rset$-derivation of $\inst_{1} \sar \emptyset$ from $\inst_{n} \sar \emptyset$. By the definition of an $\rset$-derivation, only $\seru$ rules are applied in $\prf$. We may transform $\prf$ into the chase derivation $(\inst_i, \trig_i)_{i \in [n - 1]}$ by Observation~\ref{obs:B}.
\end{proof}

\begin{customthm}{\ref{thm:simulation}}
 Let $\rset$ be a rule set. A chase derivation $(\inst_i, \trig_i)_{i \in [n]}$ witnessing $(\db, \rset) \models \exists \vec{x} \query(\vec{x})$ exists \iffi a proof in $\calc$ of $\db \sar \exists \vec{x} \query(\vec{x})$ exists.
\end{customthm}

\begin{proof} $(\Rightarrow)$ By assumption, an $\inst_{n+1}$-assignment $\assign$ exists such that $\inst_{n+1}, \mu \models \query(\vec{x})$. It follows that $\inst_{n+1} \models \exists \vec{x} \query(\vec{x})$, which implies that a proof $\prf$ of $\inst_{n+1} \sar \exists \vec{x} \query(\vec{x})$ can be constructed in $\calc$ by completeness (\thm~\ref{thm:sound-comp}). By our assumption and \lem~\ref{lem:mutual-transformation}, we obtain a derivation $\prf'$ of $\db \sar \emptyset$ from $\inst_{n+1} \sar \emptyset$. Moreover, by Observation~\ref{obs:X} we know that $\prf'$ uses only $\seru$ rules, and by repeated application of Observation~\ref{obs:A}, we may transform $\prf'$ into a new derivation $\prf''$ of $\db \sar \exists \vec{x} \query(\vec{x})$ from $\inst_{n+1} \sar \exists \vec{x} \query(\vec{x})$. By affixing $\prf$ atop $\prf''$ we obtain a proof in $\calc$ of $\db \sar \exists \vec{x} \query(\vec{x})$.

$(\Leftarrow)$ Suppose $\prf$ is a proof of $\db \sar \exists \vec{x} \query(\vec{x})$ in $\calc$. Observe that $\prf$ can only use $\seru$, $\conr$, and $\existsr$ rules as $\db$ is an instance and $\exists \vec{x} \query(\vec{x})$ is a BCQ. By Observation~\ref{lem:con-ex-permutabove-ER}, we may transform $\prf$ into a new proof which consists of two fragments: a top proof $\prf'$ of $\inst \sar \exists \vec{x} \query(\vec{x})$ consisting only of $\conr$ and $\existsr$ applications, and a bottom derivation $\prf''$ of $\db \sar \exists \vec{x} \query(\vec{x})$ from $\inst \sar \exists \vec{x} \query(\vec{x})$ consisting only of $\seru$ applications. By the soundness of $\calc$ (\thm~\ref{thm:sound-comp}), there exists an $\inst$-assignment $\assign$ such that $\assign(\query(\vec{x})) \subseteq \inst$. By Observation~\ref{obs:A} and Lemma~\ref{lem:mutual-transformation}, we can transform $\prf''$ into a chase derivation $(\inst_i, \trig_i)_{i \in [n]}$ with $\inst_{n+1} = \inst$. As $\assign(\query(\vec{x})) \subseteq \inst$, this chase derivation witnesses $(\db, \rset) \models \exists \vec{x} \query(\vec{x})$.
\end{proof}

\begin{customthm}{\ref{thm:chase-prove-equiv}}
Let $\rset$ be a rule set, $\db$ be a database, and $\exists \vec{x} q(\vec{x})$ be a BCQ. Then,
\begin{enumerate}

\item $\prove(\db \sar \exists \vec{x} q(\vec{x})) = \true$ \iffi there is an $n \in \mathbb{N}$ such that $\ch_{n}(\db,\rset) \models \exists \vec{x} q(\vec{x})$ \iffi $\ch_{\infty}(\db,\rset) \models \exists \vec{x} q(\vec{x})$;

\item If $\prove(\db \sar \exists \vec{x} q(\vec{x})) \neq \true$ , then $\inst \equiv \ch_{\infty}(\db,\rset)$ with $\inst$ the witnessing counter-model.

\end{enumerate}
\end{customthm}

\begin{proof} We note that claim (2) is straightforward as both $\inst$ and $\ch_{\infty}(\db,\rset)$ are universal models of $(\db,\rset)$. Regarding (1), we argue the first equivalence as the second equivalence is trivial. 

($\Rightarrow$) By assumption, a proof $\prf$ of $\db \sar \exists \vec{x} q(\vec{x})$ exists in $\calc$. By \thm~\ref{thm:simulation}, a chase derivation $(\inst_i, \trig_i)_{i \in [n]}$ witnessing $(\db, \rset) \models \exists \vec{x} \query(\vec{x})$ can be constructed from $\prf$. At some step $n \in \mathbb{N}$ of the chase an instance $\inst$ will be generated such that $\inst_{n+1} \subseteq \inst$, showing that $\ch_{n}(\db,\rset) \models \exists \vec{x} q(\vec{x})$. 

($\Leftarrow$) If $\ch_{n}(\db,\rset) \models \exists \vec{x} q(\vec{x})$, then $\ch_{n}(\db,\rset)$ can be linearized into a chase derivation $(\inst_i, \trig_i)_{i \in [n]}$ witnessing $(\db, \rset) \models \exists \vec{x} \query(\vec{x})$, from which a proof $\prf$ of $\db \sar \exists \vec{x} q(\vec{x})$ can be constructed in $\calc$ by \thm~\ref{thm:simulation}. By \thm~\ref{thm:correctness-proof-search}, $\prove(\db \sar \exists \vec{x} q(\vec{x})) = \true$.
\end{proof}

\end{document}